\documentclass[runningheads]{llncs}
\usepackage[T1]{fontenc}
\usepackage{newtxtext}       %
\usepackage[varvw]{newtxmath}       

\usepackage{graphicx}
\usepackage{amsmath}
\usepackage{xcolor}
\usepackage{microtype}
\usepackage{hyperref}
\usepackage{cleveref}
\usepackage{comment}
\usepackage{relsize}
\usepackage{amsfonts}
\usepackage{mathtools}
\usepackage{commath}
\usepackage[numbers]{natbib}
\usepackage{tikz}
\usepackage{algorithm}
\usepackage{algorithmic}
\usepackage{thm-restate}

\spnewtheorem{assumption}{Assumption}{\bfseries}{\itshape}

\newcommand{\shahrzad}[1]{\textcolor{blue!60!black}{\textsc{Shahrzad:} \emph{#1}}}

\newcommand{\setdiff}{\Delta}

\newcommand{\algone}{\textsc{FindTopLog}}
\newcommand{\algtwo}{\textsc{FindTopLogII}}
\newcommand{\algstar}{\textsc{STAoR}}

\newcommand{\proc}{\textsc{DiscardTail}}
\newcommand{\compare}{{\textsc{Compare}}}
\newcommand{\compareii}{{\textsc{CompareII}}}

\newcommand{\stub}{\mathsf{stub}}

\newcommand{\G}{\mathcal{G}}
\newcommand{\Ps}{\mathcal{P}}
\newcommand{\act}{\mathsf{act}}
\newcommand{\sig}{\mathsf{sg}}
\newcommand{\expprec}{\prec_*}
\newcommand{\expsucc}{\succ_*}
\newcommand{\rank}{\mathsf{rank}}
\newcommand{\N}{\mathcal N}

\begin{document}

\title{Social Learning of Ranking: Getting the Top Right despite Local
Herding}
%
%
\author{Shahrzad Haddadan \inst{1}\orcidID{0000-0002-7702-8250} \and
Jie Gao\inst{2}\orcidID{0000-0001-5083-6082} }
\authorrunning{Haddadan and Gao}
%
\institute{Rutgers Business School\\ 
100 Rockafeller Road, Piscataway, NJ 08854.\\
\texttt{shaddadan@business.rutgers.edu}
 \and
Rutgers University, Department of Computer Science \\
110 Frelinghuysen Rd, Piscataway, NJ 08854.\\
\texttt{jg1555@rutgers.edu}}
\maketitle              

\begin{abstract}

In sequential learning, a group of agents sequentially receives a noisy private signal and observes the predictions of prior agents to estimate a hidden ground truth. While extensively studied for binary signals, we formalize and analyze sequential learning when the hidden ground truth is an ordering of $n$ items, and both private signals and inferences involve pairwise comparisons.

A key focus in sequential learning is \emph{herding}, where agents disregard their private signals to follow the crowd. We show that, while local herding in our model may lead agents to make incorrect pairwise inferences, the global hidden order can still be approximated. Specifically, by aggregating all pairwise comparison outputs, the highest-ranked items can be retrieved. We investigate top-$k$ retrieval across three scenarios: (1) passive aggregation, where information is gathered from the social learning process without intervention; (2) strategic selection, where agents choose their pairwise comparisons dynamically; and (3) aggregation with stubborn nodes, where a very small set of \emph{stubborn} agents, who base their inferences solely on their private signals, join the mechanism. With query complexity as the number of agents needed for a high-probability estimation, we propose top-$k$ retrieval algorithms with query complexities sublinear in $O(n^2)$, the total number of pairs.
\end{abstract}

\section{Introduction}

Decision-making is almost always in a social setting. People make decisions by using their personal experiences and data, as well as observations of the actions and predictions of others. The topic of social learning~\cite{Golub2017-qo,Acemoglu2011-vj} concerns how decision-making occurs in a social environment and whether such decisions approach the truth. The classical model of sequential learning~\cite{Golub2017-qo,Mobius2014-oy} considers a list 
of agents trying to infer an unknown ground-truth variable. Each agent receives a private signal with a constant probability of error, and then announces her prediction. They are incentivized to correctly predict the ground truth. To do that, an agent will use her own signal as well as the predictions from agents that acted earlier. 

One main issue in sequential learning, as observed in numerous prior works, is the possibility of \emph{herding}~\cite{Banerjee1992-ra,Bikhchandani1992-rs,Welch1992-yt,Smith2000-wk,Chamley2004-or}. If the first few agents, unfortunately, have incorrect private signals, their incorrect predictions can lead all remaining agents to ``follow the herd'' and ignore their own signals. Thus, group behavior and decisions can be solely driven by the first few agents, despite abundant signals spread across the population. Numerous subsequent works have studied how herding can be mitigated by changing the network topology~\cite{Acemoglu2011-vj}, altering decision ordering~\cite{lu24enabling}, inserting a few stubborn nodes~\cite{Smith1991-sy,Sgroi2002-rz}, and so on.

The majority of the literature on social learning considers the ground truth as a simple binary variable. But social learning also occurs in more complex forms of data and knowledge. With the rapid adoption of AI models and autonomous agents, it is crucial to examine how agents share their predictions and whether their inferences and predictions can approach or converge on the truth.  In this paper, we consider social learning of an underlying ground-truth ranking of $n$ items. And we would like to examine whether the agents can learn the ground-truth ranking or approximate  top-$k$ elements in the ranking. 

Discovering the correct ranking of items with noisy measurements~\cite{Niu13stochastic}, incomplete data~\cite{Fahandar17statistical}, or bounded sampling complexity~\cite{Ren19sample} is a fundamental problem with numerous applications in e-commerce~\cite{Hayetal2017}, recommendation systems, search engines, sports analytics and tournaments~\cite{DRUGOV2020105065}, voting and epistemic social choice, and policy surveys. In this domain, the dominant assumption is that signals arrive as (noisy) pair comparisons, and the goal is to derive a ranking that is ideally close to the measurements. 

In this paper, we consider a sequential learning scenario for an underlying ranking of $n$ items $\{a_i\}_{i=1}^n$, where each agent has a private noisy observation of the comparison of some pair $a_i, a_j$. Again, the agent can hear the predictions from previous agents and will announce its prediction of the relative ordering of this pair. Our goal is to determine whether the last agent in the sequence, using all announced predictions, can infer the top-$k$ elements of the underlying ranking and how accurate this inference is. Similar to the simple setting, for any fixed pair $a_i, a_j$, among all the agents who compare this pair, herding may happen, after which agents effectively ignore their private signals, and no longer contribute to the estimation of the true underlying order. Our main question is,
\begin{quote}
\em Can the information embodied in the sequential learning process be used to approximate the true underlying order and, in particular, the top-$k$ elements? 
\end{quote}



\medskip\noindent\textbf{Related Work.}
Our setting is closely related to the literature on ``rank aggregation from pairwise comparisons''. Given an underlying true permutation, these works focus on designing algorithms to reconstruct a hidden ranking from pairwise comparisons. A fundamental assumption in these settings is that each comparison outcome is \emph{independent}. Therefore, by repeatedly querying the relationship between a pair of items, one may improve estimation accuracy. This key assumption does not generally hold in our social learning framework, where agents observe and incorporate the inferences of others into their own learning process.

Several works in the  literature address the problem of recovering the top-$k$ elements in the hidden order under the assumption of a Bradley-Terry-Luce (BTL) noise model \cite{pmlr-v28-busa-fekete13,negahban12}. 
Notably, Chen et al. \cite{pmlr-v37-chena15} developed the ``Spectral MLE'' approach, providing a rigorous sample complexity analysis and established the minimax tightness of their algorithm. Furthermore, several works have studied this problem under the more general Plackett-Luce (PL) model, where noisy observations reveal the highest-ranked item from a subset of offered alternatives \cite{pmlr-v89-saha19a,pmlr-v119-saha20b}. 
In addition to these utility-based rank models, another prominent line of work focus on the problem of learning the center of a Mallows distribution by observing pairwise comparisons \cite{lu2011learning,vitelli2018probabilistic,liu2018efficiently,2969033.2969118}.
Some works consider an \emph{active learning} paradigm, where the algorithm dynamically selects pairs to compare instead of relying on passive, random sampling \cite{pmlr-v70-mohajer17a,haddadan2026generalized,ren2020sample}.



Beyond parametric frameworks, several works study more general, non-parametric noise models, for instance many works only assume  Strong Stochastic Transitivity (SST) \cite{shah2016stochastically,lou2022active,pananjady2020worst,ren2020sample}.  Shah and Wainwright \cite{shah2018simple} provided a significant contribution by analyzing the classical Copeland count algorithm for top-$k$ item retrieval within a completely distribution-free setting. They derive sharp sample complexity bounds for estimation accuracy that depend directly on the probability gap between the $k$-th and $(k+1)$-st items. Furthermore, they establish matching lower bounds to demonstrate that no other method can achieve a strictly lower sample complexity.

A closely related line of work studies \emph{sorting without resampling}. 
Similar to the social learning paradigm, in this setting  repeated 
comparisons of the same pair return the same outcome. Braverman and Mossel \cite{BM2008} initiated this line of work, obtaining approximate rankings for any fixed positive comparison bias. Subsequent algorithms \cite{klein2011tolerant,geissmann2017sorting,geissmann2019optimal} improve computational efficiency but some impose stronger upper bounds on the error probabilities. In contrast to these works, our objective is top-$k$ retrieval rather than approximate sorting of the entire list: we seek a small candidate set containing all true top-$k$ items with high probability. Under strong stochastic transitivity, our algorithm accommodates arbitrarily small constant positive bias, without requiring an additional small-error threshold. Another important distinction of our algorithms and these works is its simplicity. Our more targeted objective, enables us to develop a   simple sampling-and-elimination procedure, avoiding the insertion-based dynamic programming of Braverman–Mossel \cite{BM2008} and the search-tree construction and merging procedures of RiffleSort \cite{geissmann2019optimal}. Despite this algorithmic simplicity, our analysis is nontrivial: the surviving candidate sets depend on previously observed persistent comparison outcomes, creating selection bias that must be controlled across successive elimination rounds.
 We will present a more detailed comparison of our work with these works in \Cref{sec:global}.


Another related line of research leverages opinion dynamics to analyze how network topology influences learning accuracy among agents seeking to infer a common ground truth \cite{mossel2014majority,olfati2007consensus,yildiz2011discrete,repec:jas,3692070.3692752}. Unlike traditional sequential learning models, these frameworks focus on iterative, network-wide opinion exchange over complex graph structures.

\medskip 
\noindent\textbf{Summary of Contributions.}
We formalize a sequential learning paradigm where the hidden ground truth is an ordering of $n$ items. Given a pair of items selected uniformly at random by nature, a learning agent receives a noisy private signal that ranks the pair and, by observing the predictions of prior agents, makes her own prediction to rank them. 
This framework generalizes the classical rank aggregation literature
to a social learning environment in which pairwise observations have complex structures.

We analyze both \emph{local} and \emph{global} herding, which we define with respect to a specific pair or across all pairs, respectively. We prove that both local and global herding emerge very early in the sequential learning process. Despite this, we demonstrate that it remains possible to aggregate potentially incorrect pairwise predictions to approximate the true hidden order.
More specifically, we develop tailored algorithms for top-$k$ retrieval across three distinct scenarios:

\begin{enumerate}
\item \emph{Passive Aggregation:} When information is collected purely by observing the un-intervened social learning process, we show how to leverage results from the sorting-without-resampling literature to recover a top-$k$ list for $k \in \Theta(\log n)$, with sample complexity (the number of agents needed) to be $O_\gamma(n^2\log n)$.

\item \emph{Sublinear algorithms using Strategic Selection:} When an algorithm can dynamically select which pairwise comparisons to present to the agents, we devise two algorithms based on simple sampling-and-elimination procedure with \emph{sublinear} query complexity  for two regimes: $k \in \Theta(\log n)$ and $k \in \Theta(\log \log n)$, achieving success probabilities of $1 - O(1/n)$ and $1 - O(1/\log n)$, respectively. We complement this by proving that this trade-off between the size of $k$ and the success probability is fundamental when restricted to regular learning agents.
\item \emph{Stubborn Participation:} We analyze a scenario where we can inject ``stubborn'' agents who base their predictions solely on their private signals, thereby breaking herding cascades. By strategically querying these agents, we design a top-$k$ retrieval algorithm that operates in the most challenging regime where $k \in \Theta(1)$.
\end{enumerate}



\section{The Model of Social Learning and Local Herding}
Consider a set of $n$ items, $\N=\{a_1,a_2,\dots, a_n\}$, that are linearly ordered. We study a sequential learning scenario in which a sequence of learning agents $\{v_t\}_{t=1}^\infty$, who are initially unaware of the true ordering, aim to learn it using both private signals from the environment and the publicly observable actions of preceding agents. At each time step $t=1,2,\dots$, nature selects a pair of distinct indices $i,j\in [n]$ uniformly at random, and agent $v_t$ is asked to compare $a_i$ and $a_j$ according to her current estimate of the underlying order. To make this decision, the agent receives a \emph{noisy private signal} from the environment and observes the history of actions made by all the preceding agents who compared the same pair. Based on this information, the agent outputs one of the two possible comparisons, $a_i\succ a_j$ or $a_j\succ a_i$, which constitutes her \emph{action}.

\medskip\noindent\textbf{Noisy Private Signals}
Without loss of generality, we assume that the true order is $a_1\succ a_2\succ \dots \succ a_n$. 
For any $a \in \mathcal{N}$, we let $\rank(a)$ denote its position in this ordering, i.e., $\rank(a_i)=i$. We denote the private signal received by agent $v_t$, conditional on nature sampling the pair $(i,j)$, as ${\sig}(v_t\mid i,j)$, and as $\sig(v_t)$ without the condition.

We assume that the noisy private signal depends only on the sampled pair $i,j$. For any $a_i,a_j\in \N$ where $i<j$ (meaning $a_i\succ a_j$), the signal probabilities are given by:
\[
\mathbb{P}\left(\sig(v_t\mid i,j)=a_i\succ a_j\right)=p_{ij}>1/2\quad \text{ and } \ \mathbb{P}\left(\sig(v_t\mid i,j )=a_i\prec a_j \right)=p_{ji}=1-p_{ij}~.
\]
We denote the $n\times n$ matrix of private signals probabilities by $\Ps$, where $\Ps(i,j)=p_{ij}$ for $i,j\in [n]$.
We assume that there exists a constant $0<\gamma<1/2$ such that, for all $i<j$ (i.e., $a_i\succ a_j$),
$
p_{ij}>\frac{1}{2}+\gamma$.
In some of our results, we impose the following natural assumption on the probability matrix $\Ps$.

\begin{assumption}[Strong Stochastic Transitivity]\label{assumption}
We assume that the private-signal probability matrix $\Ps$ satisfies for any $i<j\in [n]$, $p_{i\kappa}\ge p_{j \kappa}>0$ for all $\kappa\neq i,j$. 
In other words, 
if an item $i$ proceeds another item $j$ in the ordering, $i$ is more likely to defeat any other item $\kappa$ compared with $j$.
\end{assumption}


Assumption~\ref{assumption} is sufficiently general to capture several widely studied noise models. First, our algorithms apply to the constant-error model, where $p_{ij}=1/2+\gamma$ for every pair $i<j$. Second, they are valid in the Bradley-Terry-Luce (BTL) model, provided that the item utilities are separated by a positive constant; in this case, $\gamma$ can be chosen as the minimum pairwise comparison advantage induced by those utilities. Finally, the assumption holds under the Mallows model~\cite{mallows:biometrika57} with discrepancy parameter $\beta>0$, where one may set $\gamma=\frac{e^{\beta}-1}{2(1+e^\beta)}$.



\medskip\noindent\textbf{Agent Actions}
The action of agent $v_t$ is defined analogously as $\act(v_t\mid i,j)$ and $\act(v_t)$, respectively.
We assume that the agents are fully rational. Thus, each agent would predict the ordering of $a_i, a_j$ of highest likelihood in a Bayesian manner using her private signals and the actions of earlier agents who compared the same pair. 
When there is a tie -- i.e., the likelihood of both options are exactly the same -- we assume that the agent leans towards her own private signals.

Fix a pair $a_i, a_j$ and consider only the agents that are asked to compare them in the sequence. The $k$-th agent can see her own private signal as well as the actions (predictions) of the previous $k-1$ agents. 
If an agent can see the \emph{private} signals of the previous agents, then the best action is to report the majority vote from these private signals, whose failure probability drops as an exponential in $k$. As $k$ increases, the agents will get a very accurate prediction. Nevertheless, the challenge here is that the later agents can only see the \emph{actions} from the previous agents and must apply Bayesian inference on what is the best to report. There are several scenarios: 
\begin{itemize}
    \item The first agent will report exactly the same as the private signal. Same for the second agent, which is obvious if her private signal matches the report from the first agent. If her private signal contradicts the report from the first agent, by the tie-breaking rule, the report equals the second agent's private signal.
    \item For the third agent, there can be only two scenarios. In the first scenario, the previous two agents report opposite answers. Based on the analysis above, their private signals differ as well. Thus, the third agent will report her own private signal (which aligns with one of the previous agents). In the second case, the previous two agents report the same answer, then the third agent, using Bayesian analysis, will report the same signal as the previous agents, regardless of her own private signal. When this happens, all remaining agents will do the same thing and ignore their own signals. This is called an information cascade. 
    \item 
    This analysis generalizes to the $k$-th agent for any $k \geq 3$. If the absolute difference between the number of prior reports favoring $a_i \succ a_j$ and those favoring $a_j \succ a_i$ is at most one, the $k$-th agent will simply report their own private signal. If the margin is exactly \emph{two}, the agent disregards her private signal and follows the majority. If the margin reaches \emph{three or more}, the agent recognizes that the identical sequence of recent predictions implies a herd formation, meaning subsequent agents merely copied the crowd. Yet, even after filtering out these redundant actions, the remaining history still maintains a margin of two. Consequently, the agent is still compelled to ignore her private signal and follow the crowd. By induction, this herding cascade persists for all subsequent agents.
\end{itemize}


For any given $i, j \in [n]$, let $H_{i, j}(<t)$ denote the full sequence of actions taken by agents comparing $a_i$ and $a_j$ prior to time $t$. Furthermore, we denote the number of prior reports up to time $t$ favoring $a_i \succ a_j$ as $H_{a_i \succ a_j}(<t)$, and those favoring $a_i \prec a_j$ as $H_{a_i \prec a_j}(<t)$. Based on our previous discussion, the first step at which an agent disregards their private signal to follow the crowd occurs precisely when the absolute difference between these two tallies reaches a threshold of two, i.e.,
\begin{equation}\label{eq:herdingcond}
    \left| H_{a_i\succ a_j}(<t)- H_{a_i\prec a_j}(<t)\right| = 2~.
    \tag{Herding Condition}
\end{equation}   

Let $\tau^*(i,j)$ denote the first time step at which \ref{eq:herdingcond} is satisfied. Beyond this point, although subsequent agents continue to receive private signals and the raw sequence of observed actions $H_{i,j}(<t)$ continues to grow, the actual informational tallies freeze,  i.e., 
for any $t\geq \tau^*(i,j)$, we have,  $H_{a_i\succ a_j}(<t) = H_{a_i\succ a_j}(<\tau^*(i,j))\text{, and }H_{a_i\prec a_j}(<t) = H_{a_i\prec a_j}(<\tau^*(i,j))$.
After $\tau^*(i, j)$, herding behavior dictates that all agents comparing pair $(a_i, a_j)$  rank them in the same direction. When this occurs, we say that the edge  $(a_i, a_j)$ is \emph{fixed}, denoted by $a_i \expsucc a_j$ or $a_i \expprec a_j$, where $\expsucc$ and $\expprec$ represent the fixed preferences. We define $p^*_{ij}$ as the probability that the pair $(a_i, a_j)$ is fixed as $a_i \expsucc a_j$, and $p^*_{ji}$ as the probability that it is fixed in the opposite direction, $a_i \prec^* a_j$ ,
i.e., 
\[
p^*_{ij}\doteq \mathbb{P}(a_i\expsucc a_j) \text{, and } p^*_{ji}\doteq \mathbb{P}(a_i\expprec a_j)=1-p^*_{ij} ~.
\]
The following lemma establishes the relationship between the private signal probabilities and the probabilities of fixed edges. Furthermore, it demonstrates that any given edge will become fixed after a polynomial number of steps with high probability.

\begin{restatable}{lemma}{localherding}
\label{lem:herding:regular}
Consider a pair $i,j\in [n]$. The probability that  through  social learning the edge $(a_i,a_j)$ will be fixed as $a_i\expsucc a_j$ is:
\begin{equation}\label{eq:fixprob}
p^*_{ij}=1/\left(1+\zeta_{ij}^2\right); \text{ where } \zeta_{ij}= (1-p_{ij})/p_{ij}~.
\end{equation}
Let $\tau^*(i,j)$ be the smallest number of steps (learning agents) after which  $(a_i,a_j)$ gets fixed. 
Then we have: 
$\mathbb{E}[\tau^*(i,j)]\in\Theta(n^2)$.
\end{restatable}
\begin{proof}
Let $X_{ij}$ be a random variable defined as $X_{ij}(t)=H_{a_i\succ a_j}(<t)-H_{a_j\succ a_i}(<t)$. 
The edge will be fixed as $a_i\expsucc a_j$ the first time $X_{ij}=+2$ and it will be fixed as $a_i\expprec a_j$ the first time $X_{ij}=-2$.
Let $Y_{ij}(t)$ be a random walk starting at $Y_{ij}(0)=0$, and evolving as:
$Y_{ij}(t)=
Y_{ij}(t-1)+1$  w.p. $ p_{ij}$, and 
$Y_{ij}(t)=Y_{ij}(t-1)-1$  otherwise.
Note that $Y_{ij}$ evolves as $X_{ij}$ if we filter out the steps that $i,j$ is not sampled by nature. Thus, $p^*_{ij}$ is the probability that $Y_{ij}$ hit $+2$ before $-2$. Following the classic analysis of the gambler ruin problem 
we may conclude \cref{eq:fixprob}.

Define $\tau^*(i,j)$ the first time $X_{ij}$ hits $+2$ or $-2$. We may write 
 $\tau^*(i,j)=T_1+T_2+\dots +T_r$,
where $r$ is a random variable as the number of steps $Y_{ij}$ hits $+2$ or $-2$ and each $T_\kappa$ is the number of steps that nature takes until picking pair $i,j$ another time. Note that all $T_\kappa$ are distributed identically and they are the stopping time of a geometric random variable with probability $1/\Theta(n^2)$ of success.
Using Wald's equation we have $ \mathbb{E}[\tau^*(i,j)]=\mathbb{E}[r]\cdot\mathbb{E}[T_\kappa]$. 
From the gambler ruin classic results  
and using $\mathbb{E}[T_\kappa]=\Theta(n^2)$, we can write: 
$  \mathbb{E}[\tau^*(i,j)]=\mathbb{E}[r]\cdot\mathbb{E}[T_\kappa]\in \Theta\left(\frac{n^2}{2\max\{p_{ij},1-p_{ij}\}-1}\right)$ which is $\Theta(n^2)
$.
\end{proof}

How likely is it that a given pair $(a_i,a_j)$ is fixed in the \emph{wrong direction}? From the above lemma, we can immediately 
conclude that, for any pair $i<j\in[n]$ with true ordering $a_i \succ a_j$, if the private-signal probability $p_{ji}$ is a constant, then with constant probability, after $\Theta(n^2)$ steps, all subsequent agents comparing $(a_i,a_j)$ herd on the incorrect ordering and rank the pair as $a_i \expprec a_j$, despite their own private signals.
Thus, when considering a specific pair, our model behaves  similarly to the  classical sequential learning models with binary signals: with constant probability, herding causes all future agents to converge to and perpetuate an incorrect inference.

\section{Top-$k$ Inference Despite Herding}

Herding poses a fundamental challenge to reliable information aggregation in sequential learning systems. For a specific pair, despite abundant private signals within the agent population, herding prevents further information aggregation to reduce the uncertainty of actions. In this section, we study whether effective aggregation can still be made by harvesting the noisy comparison information to infer the top-$k$ elements in the unknown ground-truth ranking.

Given $k \in [n]$, the goal is to (approximately) recover the $k$ highest-ranked elements in the true ordering, namely $\{a_1, a_2, \dots, a_k\}$, which we denote by $\top(k)$. Note that smaller values of $k$ correspond to higher precision in identifying the highest ranked element (the \emph{winner}), whereas larger values of $k$ provide only a coarse approximation to the set of top-ranked elements.

We consider multiple settings with an increasing level of intervention and more accurate/efficient solutions. Throughout our analysis, we provide bounds on the number of agents participating in the social learning mechanism required to achieve a target outcome. We refer to this quantity as the number of steps in the sequential learning process (\Cref{sec:global}). When evaluating the performance of our algorithms, we denote this quantity as the query complexity (\Cref{sec:Alg:sublinear,sec:Alg:private}). We start (in \Cref{sec:global}) with a baseline mechanism where nature randomly picks which pair an agent will observe. Then we show (in \Cref{sec:Alg:sublinear}) that if an agent can decide which pair to observe but cannot control the error probability of the private observation, we can significantly reduce the sample complexity for inferring the top-$k$ elements. Finally, we show (in \Cref{sec:Alg:private}) that global herding can be destroyed by a very small number of stubborn nodes who refuse to participate in the Bayesian inference policy but respond solely from their private signals.



\subsection{Global Herding and Top-$k$ Inference}\label{sec:global}

In the previous section, we showed that for any pair $i,j\in[n]$, local herding on the pair $(a_i,a_j)$ occurs after $\Theta(n^2)$ steps in expectation. That is, beyond this point, additional private signals no longer help learning agents infer the true relative order of $a_i$ and $a_j$.
We define \emph{global herding} as the event in which for all pairs $i,j\in [n]$ the edge $(a_i,a_j)$ has been fixed. 
We first ask how long it takes for global herding to occur. The following theorem answers this question.



\begin{theorem}\label{thm:globalherding}
Let $T^*$ be the minimum number of steps such that after $T^*$ for all $i,j\in [n]$  the pair $(a_i,a_j)$ has been fixed. With probability at least $1-\delta$ we have: 
$
T^*\in O(n^2\ln(n/\delta))$.
\end{theorem}


\begin{proof}
For any $i,j\in [n]$, consider the agents that were asked to compare $a_i, a_j$. The learning of the order of $a_i, a_j$ behaves like a biased random walk on $\{-2,-1,0,+1,+2\}$. 
Define  $p=\min_{i<j}p_{ij}$.
If $a_i \succ a_j$ in the ground-truth order, with probability $p$ the private signal agrees with ground-truth and the random walk moves to the right ($+1$), and with probability $1-p$ moves to the left ($-1$). Whenever the random hits $2$ or $-2$, the ordering is fixed at $a_i\succ a_j$ (agreeing with ground-truth) or $a_j\succ a_i$ (disagreeing with the ground-truth), respectively. 

This biased random walk starts with position $0$. In two steps it hits the boundary $\{2, -2\}$ with probability $\alpha=p^2+(1-p)^2$ (and becomes fixed thereafter) and returns to $0$ with probability $1-\alpha=2p(1-p)$. 
Let $S_y$ be the probability that the walk starting at $0$ survives $y$ steps, i.e., not hitting the boundary. It can be shown that 
\vspace{-0.1 cm}
$$S_y=\left (1+\frac{1}{\eta}\right )\frac{\eta^y}{2}+\left (1-\frac{1}{\eta}\right )\frac{(-\eta)^y}{2},$$
where $\eta=\sqrt{2p(1-p)}$.
This is because $S_y=2p(1-p)S_{y-2}=\eta^2S_{y-2}$, with $S_0=S_1=1$. Such a recurrence relation takes the form of $S_y=A \eta^y+B (-\eta)^y$, where $A$ and $B$ are constants that can be solved by plugging in $S_0=S_1=1$.

In a total of $T$ slots, a specific walk is chosen $Y$ times, where $Y$ follows a binomial distribution $B(T, 1/m)$ and this walk is chosen with equal probability among all $m={n\choose 2}$ ones. We now compute the expected value of $S_Y$ over all possible values of $Y$, $\mathbb{E}[S_{Y}]=\sum_{y} S_y \cdot \mathbb{P}(Y=y)$. Note that this is exactly the probability that this random walk survives after $T$ slots. By linearity of expectation, 
$$\mathbb{E}[S_{Y}]=\mathbb{E}[A\eta^Y+B(-\eta)^Y]=A\cdot\mathbb{E}[\eta^Y]+B\cdot\mathbb{E}[(-\eta)^Y].$$
Using the moment generating function of the binomial distribution\footnote{For a binomial distribution $X\sim B(n, p)$, the moment generating function is $\mathbb{E}[e^{tX}]=(1-p+pe^t)^n$.}, we can show that for any $x$ $$\mathbb{E}[x^{Y}]=\left ( \frac{x}{m}+1-\frac{1}{m}\right)^T.$$
Thus, 
$$\mathbb{E}[S_{Y}]=\frac{1}{2}\left (1+\frac{1}{\eta}\right )\left (1-\frac{1-\eta}{m}\right )^T+\frac{1}{2}\left (1-\frac{1}{\eta}\right )\left (1-\frac{1+\eta}{m}\right )^T\leq \frac{1}{2}\left (1+\frac{1}{\eta}\right )\left (1-\frac{1-\eta}{m}\right )^T.$$
The last inequality is because $\eta<1$ and the second term is negative.
Thus, the right-hand side is an upper bound of the probability that after $T$ slots, this walk is not fixed. We use union bound to show that we want $$m \cdot \frac{1}{2}\left (1+\frac{1}{\eta}\right )\left (1-\frac{1-\eta}{m}\right )^T\leq \delta,
$$
where $m={n\choose 2}$. This solves $T=O(n^2\ln (n/\delta))$ to be sufficient.
\end{proof}

We remark that after $T=\Theta(n^2\ln n)$, all pairs are fixed with high probability (taking $\delta=1/n$).  This bound cannot be reduced further, since $\Omega(n^2\ln n)$ agents are needed just to see all pairs $a_i, a_j$, by the coupon collector problem~\cite{Mitzenmacher2017}. 

The outcome of global herding is a complete directed random graph in which the orientation of the edge between each pair $(a_i, a_j)$ is independently sampled according to the probabilities $\{p^*_{ij}, 1 - p^*_{ij}\}$. Beyond this point, the graph becomes fixed and does not change as the process evolves further. 
The problem of recovering the underlying ranking from such a fixed random graph has been a central focus in the ``sorting without resampling'' literature.  Baverman and Mossel \cite{BM2008} were the first to study this problem. They seek to find an underlying true order,  given a complete random graph $\G$ with vertex set $[n]$, where for any $i\succ j$ the direction of each edge $(i, j)$ is sampled as $i \succ j$ with  probability at least $1/2 + \gamma$ and in the opposite direction with probability at most $1/2 - \gamma$ for constant $\gamma$. 

This model aligns with ours. 
In this case, the expected number of edges sampled in the incorrect direction constitutes a constant fraction of all edges. While this high error rate seemingly presents a  challenge for the reconstruction of the true order, Braverman and  Mossel \cite{BM2008} demonstrated the existence of a permutation $\sigma^*$, computable in $O(n^{2\gamma^{-4}})$ time and sampling  $O_\gamma(n\log n)$ edges from $\G$, such that the \emph{maximum dislocation} in $\sigma^*$ is bounded by $\Theta(\log n)$, and 
the \emph{total dislocation} of $\sigma^*$ is bounded by $\Theta(n)$, i.e, with  probability at least $1-o(1/n)$:
\vspace{-0.1 cm}
\begin{equation}\label{eq:max:dislocation}
\forall a \in \N, |\text{rank}(a) - \sigma^*(a)| \leq \Theta(\log n), ~ \text{and}\quad  \sum_{a\in \N}|\text{rank}(a) - \sigma^*(a)| \leq \Theta(n).
\end{equation}

Geissmann et al. \cite{geissmann2017sorting} considered the same problem and error model, improving the runtime for finding $\sigma^*$ to $O(n^2)$. 
Later on, they proposed \textsc{RiffleSort}, an adaptive algorithm with runtime and sample complexity  
$O(n\log n)$, but requiring that  the error probabilities are always below 1/16 \cite{ geissmann2019optimal}.

\vspace{-0.2 cm}
\subsection{Sublinear Query for Top-$k$  Discovery}\label{sec:Alg:sublinear}

It is worth noting that, once global herding occurs, existing algorithms from the sorting-without-resampling literature, such as \cite{geissmann2017sorting}, can efficiently recover the top-$k$ elements for $k\in\Theta(\log n)$ in polynomial time. The challenge is that, without intervening in the social learning process, this process  requires querying $O_\gamma(n^2\log n)$ agents, as shown in \Cref{thm:globalherding}. \textsc{RiffleSort} \cite{geissmann2019optimal} can be interleaved in the social learning process to achieve query and runtime complexity $O_\gamma(n\log n)$. However, it requires the error probabilities for all edges to remain below $1/16$.
The objective of this section is to improve upon these guarantees by interleaving the execution of our algorithms with the sequential learning process itself.
Our goal is to design adaptive polynomial time algorithms for top-$k$ retrieval with the number of learning agents being \emph{sublinear} in the total number of edges in $\G$, i.e., $o(n^2)$ by \emph{only} assuming the strong stochastic transitivity (Assumption \ref{assumption}).



To this end, we develop two algorithms with sublinear query complexity for $\top(k)$ retrieval using sequential learning agents, in the regimes $k\in\Theta(\log n)$ and $k\in\Theta(\log\log n)$.  The algorithms adaptively select indices $i,j\in[n]$ and present the pair $(a_i,a_j)$ to the next learning agent, whose action is to compare $a_i$ and $a_j$. Our analysis relies solely on Assumption~\ref{assumption}, placing no further restrictions on the error probability. Unlike \textsc{RiffleSort}~\cite{geissmann2019optimal}, which relies on complex search-tree constructions and merging procedures, our algorithms are conceptually simple, driven entirely by repeating a  sampling and elimination procedure.
We  show that the two algorithms exhibit different trade-offs between the success probability and the achievable size of $k$. \Cref{thm:fistalg,thm:secondalg} formalize these trade-offs. Furthermore, \Cref{lem:lowerbound1} provides a lower bound analysis, demonstrating that this trade-off is inherent to the problem.

\medskip 

We now propose \algone, our adaptive algorithm for finding the top-$k$ elements when $k\in \Theta(\log n)$.
The pseudocode of  \algone\ is presented in 
\cref{alg:first}. It starts with the initial set $\Gamma_0=\N$ and iteratively calls  \proc\ (\Cref{alg:dividproc}) to obtain $\Gamma_{r+1}=\proc(\Gamma_r)$ for $r=0,1,2\dots, \mathbf{r}-1$ until reaching $\Gamma_\mathbf{r}$ with cardinality less than $L\in \Theta_\gamma\left( \log n\right)$.  The algorithm will terminate after  outputting $\Gamma_\mathbf{r}$. The pseudocode  of \proc\ is presented in \Cref{alg:dividproc}. At the $r$th run, it receives as input $\Gamma_{r-1}$ and estimates the rank of each element  $a_i\in \Gamma_{r-1}$ by collecting  a set $N_i$ of $D_{r-1}$ other elements sampled independently, uniformly at random and with replacement from $\Gamma_{r-1}$ (using \Cref{alg:compare1}). Formally, for any set $S\subseteq \N$ 
let $\rank(a_i,S)$ be defined as:
\[
\rank(a_i,S)=\left|\{a'\in S\mid a'\succ a_i\}\right|~.\footnote{Throughout this proof, item ranks are $0$-indexed ($\{0, 1, 2, \dots\}$) for mathematical convenience, contrasting with the $1$-based convention used in the rest of the paper.}
\]

\begin{minipage}{0.44\textwidth}

\begin{algorithm}[H]
\caption{\compare}
\label{alg:compare1}
{\small
\begin{algorithmic}[1]
\STATE \textbf{Input:} $i,j\in \N$
\STATE $s=H_{a_i\prec a_j}(<t)-H_{a_i\succ a_j}(<t)$
\REPEAT
    \STATE query the next learning agent, $v_t$, to compare $a_i,a_j$
    \IF{$\sig(v_t)=a_i\prec a_j$}
        \STATE $s++$
    \ELSE
        \STATE $s--$
    \ENDIF

\UNTIL{$s=+2$ or $s=-2$}
\IF{$s>0$}
    \RETURN $a_i\expprec a_j$
\ELSE
    \RETURN $a_i\expsucc a_j$
\ENDIF
\end{algorithmic}
}
\end{algorithm}

\end{minipage}
\hfill
\begin{minipage}{0.53\textwidth}

\begin{algorithm}[H]
\caption{\compareii}
\label{alg:compare2}
{\small
\begin{algorithmic}[1]
\STATE \textbf{Input:} $i,j\in \N$, number of repetitions  $\theta\in \mathbb{N}$
\STATE $s=0$
\FOR{$\kappa=1:\theta$}
    \STATE query the next \textbf{stubborn} learning agent, $v_t$, to compare $a_i,a_j$
    \IF{$\sig(v_t)=a_i\prec a_j$}
        \STATE $s++$
    \ELSE
        \STATE $s--$
    \ENDIF
\ENDFOR
\IF{$s>0$}
    \RETURN $a_i\prec a_j$
\ELSE
    \RETURN $a_i\succ a_j$
\ENDIF
\end{algorithmic}
}
\end{algorithm}
\end{minipage}

\bigskip
We define two random variable for estimating the rank of $a_i$ in $S$,  $\hat{R}_i(S)$ and its normalized  rank $\hat{\mu}_i$   as:
\[
\hat{R}_i(S)=\left|\{a'\in N_i\mid  a_i\succ^* a'\}\right|, \text{ and } ~\hat{\mu}_i(S)=\hat{R}_i(S)/|N_i|~.
\]

Let $|\Gamma_{r-1}|=m_{r-1}$, at the $r$th run, \proc\ calculates the estimated ranks $\hat{R}_i(\Gamma_{r-1})$ for any $a_i\in \Gamma_{r-1}$ and discards $\lfloor m_{r-1}/3\rfloor$ elements with lowest estimated ranks (lowest values of $\hat{R}_i(\Gamma_{r-1})$) by breaking the ties arbitrarily. 

\paragraph{\textbf{A road-map to the proof of correctness of \algone.}}
Given a set $S\subseteq \N$ with cardinality $|S|=m$, let's define for any 
$a_i\in S$, 
\[\mu_i(S)=\frac{\sum_{j\in S\setminus \{i\}}p^*_{ij}}{m-1}~.\]
The following lemma is proven in the full paper:

\begin{restatable}{lemma}{mugap}\label{lem:mugap}
Let $S$ be an arbitrary subset of $\N$ with $|S|=m$. Under Strong Stochastic Transitivity assumption, for any $a_i,a_j$ such that $\rank(a_i,S)<\rank(a_j,S)$, we have:
\[
\mu_i(S)-\mu_j(S)\geq \frac{2\gamma}{m-1}\left(
\rank(a_j,S)- \rank(a_i,S)
\right)~.
\]
\end{restatable}

Given the above lemma, if $S$ is a \emph{fixed set}, one can use the Hoeffding bound and use  $\mathbb{E}[\hat{\mu}_i(S)]=\mu_i(S)$ to show a gap between the estimated ranks of high and low ranked elements in $S$ and  that \proc\ does not remove any of the elements in $\top(m_{r-1}/3)$. The technical challenge is that, the set $\Gamma_{r-1}$ is selected \emph{adaptively}, and some of the edges in it have already been sampled in previous runs of \proc\ and have known orientations. This creates a deviation of $\mathbb{E}[\hat{\mu}_i(S)]$ from $\mu_i(S)$. Formally, for any $a_i\in S$, we define
\[
E_{ij}=\mathbf{1}(a_i\succ^* a_j)-p^*_{ij}~, \text{ and}\quad  E_i(S)= \frac{\sum_{j\in S\setminus\{i\} }E_{ij}}{m-1}
~.\]
Therefore,  we can write: 
\begin{equation}\label{eq:correctionterm}
  \mathbb{E}\left[\hat{\mu}_i(S)\right]= \mu_i(S) +  \underbrace{E_i(S)}_{\text{correction term}}~.  
\end{equation}
Note that \cref{eq:correctionterm} is valid for any adaptive set $S$, conditioning on the history of the sampled graph.

\begin{minipage}{0.48\textwidth}
\begin{algorithm}[H]
\caption{\algone$(\N, \gamma)$}\label{alg:first}
{\small
\begin{algorithmic}[1]
\STATE \textbf{Input:} set $\N$, bias $\gamma$
\STATE $L=(8*320^3/\gamma^5) \log n$
\STATE $\Gamma=\N$
\REPEAT
\STATE $\Gamma=\proc(\Gamma,1/n, \gamma)$
\UNTIL{$|\Gamma|\leq L$}
\RETURN $\Gamma$
\end{algorithmic}
}
\end{algorithm}
\end{minipage}
\hfill
\begin{minipage}{0.48\textwidth}
\begin{algorithm}[H]
\caption{\algtwo$(\N, \gamma)$}\label{alg:second}
{\small
\begin{algorithmic}[1]
\STATE \textbf{Input:} set $\N$, bias $\gamma$
\STATE $L=
\frac{8*320^3}{\gamma^5}
\log(\log n)+ \frac{8*320^3}{\gamma^5}\log(\frac{8*320^3}{\gamma^5})
$
\STATE $\Gamma=\algone(\N)$
\REPEAT
\STATE $\Gamma=\proc(\Gamma,1/\log(n),\gamma)$
\UNTIL{$|\Gamma|\leq L$}
\RETURN $\Gamma$
\end{algorithmic}
}
\end{algorithm}
\end{minipage}

\begin{algorithm}
{\small
\caption{\proc$(S, \delta, \gamma)$}\label{alg:dividproc}
\begin{algorithmic}[1]
\STATE \textbf{Input:} set $S$, success probability $\delta$, and bias $\gamma$.
\STATE $m=|S|$
\STATE $C_\gamma=4*320^3/\gamma^5$
\STATE $D=C_\gamma\log (m\cdot \delta^{-1})$
\FOR{$i\in S$}
\STATE $\hat{R}_i=0$
\STATE $N_i\leftarrow$ $D$ uniform random samples from $S\setminus\{i\}$
\FOR{$a_j\in N_i$}
\STATE \textbf{if} $ \compare(i,j)=a_j\expprec a_i$ \textbf{then}, $\hat{R}_i++$
\ENDFOR 
\ENDFOR
\STATE Sort $S$ w.r.t $\hat{R}_i$ and break ties arbitrarily 
\STATE discard $\lfloor m/3\rfloor$ elements of $S$ with lowest ranks. Let the remaining elements be $\Gamma$
\RETURN $\Gamma$
\end{algorithmic}
}
\end{algorithm}

Consider the $r$th run of \proc.   Note that the correction terms appear as    $E_i(\Gamma_r)$ in our analysis and $\Gamma_r$ is obtained in an adaptive manner, thus the expectation of $E_i(\Gamma_r)$ is not a fixed value. Therefore,   concentration bounds cannot be trivially be applied to bound the correction term. To address this challenge we show (using induction) that for any $r\in\{1,2,\dots ,\mathbf{r}\}$, the set $\Gamma_r$ is not too different from the deterministic set $\top(m_r)$. Using this observation, we can employ Hoeffding bound to $E_i(\top(m_r))$, with expectation zero, which constitutes a key step in the proof. The following theorem states this formally, and the proof is presented in the full paper.

\begin{restatable}{theorem}{inductive}\label{inductiveproof}
Let $\Gamma_0=\N$ and for $r\in \{0,1,2,\dots, \mathbf{r}-1\}$ we have: $\Gamma_{r+1}=\proc(\Gamma_r)$. For $r\in \{0,1,\dots, \mathbf{r}\}$ we denote $|\Gamma_r|=m_r$. Given $\delta\geq n^{-2}$, and $m_{r}\geq 4*320^3/\gamma^5\log n$, at round $r$ it holds that:
\begin{align*}
 1.\quad & \top(\lfloor\frac{11}{16}m_r\rfloor)\subseteq \Gamma_r\subseteq \top( 2 m_r)~.\\
 2.\quad & \big|  \Gamma_r\setdiff\top(m_r) \big|\leq \frac{\gamma}{80}(m_r-1)~.
\end{align*}
with a failure probability at most  
\[
\frac{9\delta}{n}\sum_{i=0}^{r} (\frac{3}{2})^i~.
\]
\end{restatable}

\begin{proof}
We prove this theorem using induction on $r$. For $r=0$,  both statements hold trivially with probability one since $\Gamma_0=\N=\top(n)$. 
For the inductive step assume $(1)$ and $(2)$ to be true for $r$ and we show them for $r+1$ through the following  steps:

\smallskip 

\noindent 
\emph{Step 1.} If  (1) and (2) hold for $S=\Gamma_r$ we have:
\begin{equation}\label{concentratedmean}
 \mathbb{P}\left( \forall i \in S,~ \big| \hat{\mu}_i(S)-\mu_i(S)\big|  \leq \gamma/5\right)\geq 1- 5\delta/m_r~. 
\end{equation}
Note that the above bound is shown for adaptively selected $S$ as long as  (2) holds. 

To show \cref{concentratedmean}, we consider an arbitrary $a_i\in S$ and write:
\begin{align}\label{eq:decompose}
\big|
\hat{\mu}_i(S)-\mu_i(S)
\big|&\leq 
\big|
\hat{\mu}_i(S)-\mathbb{E}[\hat{\mu_i}(S)]
\big|+ \big|
\mathbb{E}[\hat{\mu_i}(S)]-\mu_i(S)
\big|\\
&=\big|
\hat{\mu}_i(S)-\mathbb{E}[\hat{\mu_i}(S)]
\big|+\big|E_i(S)\big|~
\end{align}
Note that the first term can be bounded through a straightforward application of the Hoeffding bound as follows: 
\begin{align}\label{firstterm}
&\mathbb{P}\left(\big|
\hat{\mu}_i(S)-\mathbb{E}[\hat{\mu}_i(S)]\big|
\geq \frac{7}{40}\cdot \gamma
\right)\\
\leq & 2\exp\left(
-2 D_r(\frac{7\cdot \gamma}{40})^2
\right)\leq 2 \exp\left(-\frac{2*49}{1600\cdot \gamma^2}C_\gamma\log \frac{m_r}{\delta}\right)\\
\leq & 2\exp\left(-2\log \frac{m_r}{\delta}\right)\leq 2\delta/m^2_r~.
\end{align}

In order to bound $E_i(S)$, we first bound $E_i\left(\top(m_r)\right)$. Since $\top(m_r)$ is a fixed set we may employ  the Hoeffding bound on the average of $m_r$ independent, zero mean random variables. 

\[
\mathbb{P}\left(
\big|E_{i}\left(\top(m_r)\right)\big|\geq 
\gamma/80
\right)\leq 2\exp\left(-
(\gamma^2/6400)\cdot  m_r
\right)~.
\]
Note that 
\[
E_i(S)=\frac{\sum_{a_j\in S\setminus\{a_i\}}E_{ij}}{m_r-1}\leq  \frac{\sum_{a_j\in \top(m_r)}E_{ij}}{m_r-1}+ \frac{ S\setdiff \top(m_r) }{m_r-1}
\]
Using the above and (2) for all $a_i\in S$ we write:  
\begin{equation}\label{secondterm}
    \big| E_i(S)\big|\leq \big| E_i\left(\top(m_r)\right)\big|+ \big| S\setdiff \top(m_r)\big|/(m_r-1)\leq  \frac{\gamma}{80}+ \frac{\gamma}{80}= \gamma/40~.
\end{equation}

Plugging in \cref{firstterm} and \cref{secondterm} in \cref{eq:decompose} we obtain:  
\[
\big|
\hat{\mu}_i(S)-\mu_i(S)
\big|\leq \frac{7}{40}\gamma+ \frac{1}{40}\gamma= \gamma/5~.\] 

Therefore, if (2) holds for $S=\Gamma_r$  for $a_i\in S$ and a large enough $m_r$ we have: 
\[
\mathbb{P}\left(\big|
\hat{\mu}_i(S)-\mu_i(S)
\big|\geq \gamma/5\right)\leq  2\exp \left(-\frac{\gamma^2}{6400}\cdot m_r\right)+2\delta/m^2_r~.
\]
Taking a union bound over all $a_i\in \top(2m_r)$ and using $S\subseteq \top(2m_r)$ (statement 1 of induction hypothesis), we obtain:
\[
\mathbb{P}\left(
\exists a_i\in \top(2m_r), 
\big|
\hat{\mu}_i(S)-\mu_i(S)
\big|\geq \gamma/5\right)\leq  2m_r\cdot \left[2\exp \left(-\frac{\gamma^2}{6400}\cdot m_r\right)\right]+4\delta/m_r~.
\]

Using $n\geq m_r\geq \frac{8*320^3}{\gamma^5}\log n$ and $\delta\geq 1/n^2$:
\begin{align}\label{firstbound}
   4m_r\exp\left(-\frac{\gamma^2}{6400}\cdot m_r\right)&\leq 4m_r\exp\left(-\frac{\gamma^2}{6400}\cdot \frac{8*320^3}{\gamma^5}\log n\right) \nonumber\\
   &
   \leq 4m_r\exp\left(-\frac{8*512}{\gamma^3}\log n\right) \nonumber
   \\
   &\leq 4m_r\exp\left(-512\log n\right) \leq \frac{4m_r}{n^{512}}\leq \delta/m_r~.
\end{align}
Putting together, we have:
\[
\mathbb{P}\left(\exists a_i\in S,~\big|
\hat{\mu}_i(S)-\mu_i(S)
\big|\geq \gamma/5\right)\leq 5\delta/m_r~.
\]

\smallskip 

\noindent
\emph{Step 2.} Assume (1) for round $r$.  We have:
\begin{equation}\label{eq:step2}
\mathbb{P}\left(\big|
\top(m_{r+1})\setminus \Gamma_{r+1}\big|\leq \frac{\gamma}{320}\cdot m_r
\right)\geq 1- 4\delta/m_r
\end{equation}

We show Step 2 in \Cref{lem:setdif}.
Note that from \Cref{eq:step2}, the statement 2 of the inductive step can be proved directly as follows:
\begin{align*}
&\big|\Gamma_{r+1}\setdiff \top(m_{r+1})\big|= \big| \top(m_{r+1})\setminus \Gamma_{r+1}\big|+\big| \Gamma_{r+1}\setminus \top(m_{r+1})\big|=2 \big| \top(m_{r+1})\setminus \Gamma_{r+1}\big|\\
\leq & \frac{\gamma}{160}\cdot m_r\leq \frac{\gamma}{160}\cdot \frac{3}{2}\cdot m_{r+1}\leq 
\frac{\gamma(m_{r+1}-1)}{80}~.
\end{align*}
The first equation is because $\big| \Gamma_{r+1}\big|=m_{r+1}$.

We now finish proving the inductive step by showing statement 1 holds for $r+1$. 
Knowing $m_{r+1}=\lceil 2m_r /3\rceil $, we show the following stronger result: 
\[
\top(\lfloor \frac{11}{16} m_{r+1}\rfloor)\subseteq \Gamma_{r+1}\quad \& \quad \Gamma_{r+1}\subseteq \top\left( m_{r} \right)~.
\]
 First note that for large enough $m_{r+1}$
 \[\top(\lfloor \frac{11}{16} m_{r+1}\rfloor )\subseteq \top(m_{r+1})\subseteq \top(\lceil \frac{11}{16}m_r\rceil)\subseteq \Gamma_r~. \]

Consider  $a_i\in \top(\lfloor \frac{11}{16} m_{r+1}\rfloor)$, which implies  $a_i\in \Gamma_{r}$. We show that  \proc\ will not remove $a_i$ from $\Gamma_r$.
For  any element $a_j$ in $\Gamma_r$ but  not in $\top(m_{r+1})$, since $\top(m_{r+1})\subseteq \Gamma_r$ we have that $\rank(a_j,\Gamma_r)> m_{r+1}$. Furthermore, for $a_i\in \top(\lfloor \frac{11}{16} m_{r+1}\rfloor)$ we have   $\rank(a_i,\Gamma_r)\leq  \frac{11}{16} m_{r+1} $.  Therefore,  for any $a_i\in \top(\lceil \frac{11}{16}m_{r+1}\rceil ), a_j\in \Gamma_r\setminus \top(m_{r+1})$ we have:
\[
\rank(a_j,\Gamma_r)-\rank(a_i,\Gamma_r)\geq \frac{5}{16}m_{r+1}\geq \frac{5}{24}m_r~\underset{\Cref{lem:mugap}}
 {\implies}\quad 
  \mu_i(\Gamma_r)-\mu_j(\Gamma_r)\geq \frac{5}{12}\gamma~.
\]

Using \cref{concentratedmean} we conclude that with probability at least $1-2\delta/m_r$ we have:
\[
a_i\in \top(\lfloor \frac{11}{16} m_{r+1}\rfloor), a_j\in \Gamma_r\setminus \top(m_{r+1})\implies \hat{\mu}_i-\hat{\mu}_j\geq 5\gamma/12 - 2\gamma/5>0~.
\]

Note that \proc\ preserves $m_{r+1}$ elements of $\Gamma_{r}$. We have that $a_i$ beats any item outside $\top(m_{r+1})$, and there are only $m_{r+1}-1$ elements other than $a_i$ inside $\top(m_{r+1})$. Therefore, with probability at least $1-2 \delta/m_r$ \proc\ will not remove $a_i$ from $\Gamma_r$.

\medskip 

We now show that any element outside $\top(m_r)$ will be removed with high probability by \proc. 
Note that any $a_j\in \Gamma_r\setminus \top(m_r)$ satisfies $\rank(a_j,\Gamma_r)> m_r-\vert \Gamma_r \setminus\top(m_r) \vert $.
Using statement 2 we have that 
\[
\vert \Gamma_r \setminus\top(m_r) \vert  \leq \vert \Gamma_r \setdiff\top(m_r)\vert /2 \leq  (\gamma/160)\cdot (m_r-1).
\]

Therefore, $\rank(a_j,\Gamma_r)\geq  m_r(1-\gamma/160) $. Consider any $a_i\in \top(m_{r+1})$. Since $m_{r+1}\leq \lceil 2/3 m_r\rceil $, we have that $\rank(a_i,\Gamma_r)\leq 2m_r/3$. Thus,
\[
\rank(a_j,\Gamma_r)-\rank(a_i,\Gamma_r)\geq (1/3-\gamma/160)m_r~\underset{\substack{\Cref{lem:mugap}\\
 \gamma/160\leq 1/320
 }}
 {\implies}\quad 
  \mu_i(\Gamma_r)-\mu_j(\Gamma_r)\geq \frac{2}{5}\gamma~.
\]

Using \cref{concentratedmean} we conclude that with probability at least $1-2\delta/m_r$ the estimated rank of any $a_i\in \top(m_{r+1})$ will be larger than the estimated rank of any  $a_j\in \Gamma_{r}\setminus\top(m_{r})$. Note that \proc\ preserves only $m_{r+1}$ elements and $\vert \top(m_{r+1})\vert=m_{r+1}$. Therefore, $a_j$ will be deleted.

By induction hypothesis, the probability that (1) or (2) do not hold is bounded by 
\[
\frac{9\delta}{n}\sum_{i=0}^{r-1}(3/2)^i.
\]

From \cref{concentratedmean} and \cref{eq:step2} the probability of failure at this round is bounded by $3\delta/m_r+4\delta/m_r$.
Since $m_r\leq \lceil(2/3)^rn\rceil $, putting all  together we get that the probability of failure by the end of this round is bounded by 
\[
\frac{9\delta}{n}\sum_{i=0}^{r-1}(3/2)^i+ 
9\delta/m_r\leq \frac{9\delta}{n}\sum_{i=0}^{r}(3/2)^i~.
\]

\end{proof}
We now proceed to the proof of Step 2. Here, we encounter the same technical obstacle: establishing concentration for an adaptively determined set whose expectation is unknown. We resolve this by analyzing a deterministic candidate set and taking a union bound over all realizable instances. Crucially, statement (1) of the inductive step bounds the range of input sets to \proc{}, allowing us to control the number of candidate sets. Specifically, we show the following lemma with proof deferred to the appendix:

\begin{restatable}{lemma}{setdiflemma}
\label{lem:setdif}
For any $S\subseteq \N$, let  $S'=\proc(S)$. Consider all $S$ with $|S|= m_r$ satisfying $\top(\lfloor \frac{11}{16}m_r\rfloor)\subseteq S\subseteq \top(2m_r)$. Let $\delta\geq n^{-2}$, and $m_{r}\geq 4*320^3/\gamma^5\log n$,  we have that:
\begin{equation}
\mathbb{P}\left(\forall S\subseteq \top(2 m_r):  \big|
\top(m_{r+1})\setminus S'\big|\leq \frac{\gamma}{320}\cdot m_r
\right)\geq 1- 4\delta/m_r
\end{equation}
\end{restatable}

\medskip 

Our first algorithm, \algone, is designed for the regime $k\in\Theta(\log n)$ and achieves success probability $1-O(1/n)$.

\begin{restatable}{theorem}{lognalg}
\label{thm:fistalg}
Algorithm \algone\ receives as input $\N$ and bias parameter $\gamma$, and outputs a set $\Gamma$ of size $m$ 
satisfying,
\[
m\in O_{\gamma}\left(
\log n
\right), ~\text{and }~
\mathbb{P}\left(\exists a\in \top\left(\lfloor m/3\rfloor\right)\text{ such that } a\notin \Gamma \right)\leq 1/n~.
\]
Equivalently, with probability at least $1-1/n$, $
\top\left(\lfloor m/3\rfloor\right)\subseteq \Gamma.
$
Let $t$ denote the total number of agents used by \algone, we have: 
$\mathbb{E}[t]\in O_{\gamma}\left(n \log n\right).$
\end{restatable}

\begin{proof}
The proof follows directly from \cref{inductiveproof}, which in fact shows a stronger statement. Let $\mathbf{r}$ be the final round of applying \proc. Set $r=\mathbf{r}$. 
Note that $\delta=1/n$ in \algone. Taking $C_\gamma=8*320^3/\gamma^5$, we have, 
\begin{align*}
    \mathbb{P}\left(
\top(\lceil\frac{11}{16} m_{\mathbf{r}}\rceil )\not\subset \Gamma_\mathbf{r}
\right)&\leq 9\delta\sum_{r=0}^\mathbf{r}\frac{1}{m_r}& \text{Statement 1 of \cref{inductiveproof} }\\
&\leq \frac{9}{n}\cdot \frac{\log n}{C_\gamma\log n }& \mathbf{r}\leq \log n, \delta=1/n, m_r\geq C_\gamma \log n\\
&\leq 1/n
\end{align*}

Therefore, the first part follows from $\top(\lfloor m/3\rfloor)\subseteq \top(\lfloor 11m/16\rfloor)$.

To see the second part of the theorem which bounds the number of learning agents, note that for any sampled pair to get fixed, we need to query a constant number of agents in expectation. To simplify the analysis we assume that this number is  fixed and bound the number of sampled pairs:
\begin{align*}
\text{total number of sampled pairs}&\leq \sum_{r=0}^{\log n}  m_rD_r\leq  \frac{8*320^3}{\gamma^5}\sum_{r=0}^{\log n} m_r \log (nm_r) \\
& 
\leq   
\frac{8*320^3}{\gamma^5}\sum_{r=0}^{\log n} \frac{2}{3}^r n \log\left(
n^2 (2/3)^r
\right)\\
&\leq \frac{8*320^3}{\gamma^5} n \sum_{r=0}^{\log n} (2/3)^r [2\log n + r\log (2/3)]\\
&\leq  \frac{24*320^3}{\gamma^5}n\log n \sum_{i=0}^\infty (2/3)^r\\
&\in O_{\gamma}(n\log n )~.
\end{align*}
\end{proof}


\subsubsection{Retrieving $\top(m)$ in $m\in O(\log \log n)$ regime } 
To retrieve the top-$m$ elements in the regime $m \in \Theta(\log \log n)$, we propose \algtwo{}. The algorithm first invokes \algone{} to obtain an initial set of size $\Theta(\log n)$. It then repeatedly executes \proc{} in an iterative fashion, analogous to \algone{}. Within \proc{}, the failure probability is governed by the parameter $\delta$, where the sample size for estimating item ranks satisfies $D_r \in \Theta(\log(m_r \delta^{-1}))$. Because the input set shrinks with each iteration, we configure $\delta = 1/\log n$. The formal guarantee is stated below, with the complete proof provided in the appendix:

\begin{restatable}{theorem}{thmalgtwo}\label{thm:secondalg}
Algorithm \algtwo\ receives as input $\N$ and bias parameter $\gamma$, and outputs a set $\Gamma$ of size $m$ 
satisfying,
\[
m\in O_{\gamma}\left(
\log \log n
\right), ~\text{and }~
\mathbb{P}\left(\exists a\in \top\left(\lfloor m/3\rfloor\right)\text{ such that } a\notin \Gamma \right)\leq 1/\log n ~.
\]
Equivalently, with probability at least $1-1/\log n$, $
\top\left(\lfloor m/3\rfloor\right)\subseteq \Gamma.
$

\medskip 

\noindent 
Let $t$ denote the total number of agents used by \algone, 
we have: 
\[
\mathbb{E}[t]\in O_{\gamma}\left(n \log n\right)+ O_{\gamma}\left(\log n \log \log n\right).
\]
\end{restatable}

\begin{remark}
In \Cref{thm:fistalg,thm:secondalg}, we analyzed the performance of our algorithms with respect to their query complexity. It is worth noting that the runtime complexity of both algorithms is governed by essentially the same bounds. Indeed, due to their simple design, the algorithms incur minimal computational overhead and, on average, require only a constant amount of computation per agent.
\end{remark}

\subsubsection{Lower bounds and the impossibility of retrieving $\top(m)$ for $m\in \Theta(1)$}

We now present our lower bound result
demonstrating that  the above tradeoff between the probability of success and the size of $k$ for top-$k$ recovery is tight. Our construction should remind the reader of the lower bound result presented in \cite{geissmann2017sorting} which we here tailor to our  problems. 

\smallskip

To show the lower bounds  we focus on specific family of private signal probability matrices which follow SST: Given constant $0<\gamma<1/2$, for a given hidden permutation $\sigma$, we define $\Ps(\sigma)$ as follows: for all $i<j$, let $p_{\sigma(i)\sigma(j)}=1/2+\gamma$ and $p_{\sigma(j)\sigma(i)}=1/2-\gamma$. 
Throughout this section, when we refer to an \emph{algorithm} we mean an algorithm as described bellow:


Let $Alg$ be  an arbitrary algorithm 
which without knowing  $\Ps(\sigma)$ or $\sigma$ (but maybe knowing $\gamma$), samples 
complete random graph $\G$  generated by the sequential learning process w.r.t private signal probability matrix $\Ps(\sigma)$ and processes $\G$ with the goal to retrieve the top $m$ elements of the hidden permutation $\sigma$, i.e.,  $\top_{\sigma}(m)=\{\sigma(1),\sigma(2),\dots , \sigma(m)\}$.   The following lemma shows a tradeoff between size of $m$ and the correctness probability of any such $Alg$.

\begin{lemma}
\label{lem:lowerbound1}
For any given $0<\gamma<1/2$, consider an arbitrary algorithm $Alg$ as described above with output size   $m$. 
There exists constant $c_\gamma<1$ and  a hidden permutation $\sigma$  such that    with probability at least $c_\gamma^{m}/(m+1)$ the output of $Alg$ does not contain $\sigma(1)$.
\end{lemma}

\begin{proof}
We consider a deterministic version of $Alg$, which is  for a given  sampled graph $\G$, we fix all the other  
 the random bits used  during the performance of the algorithm when it processes $\G$. Therefore, the outcome of the algorithm on a fixed sampled $\G$ is the same. 

Let's define $\sigma_1=a_1\succ a_2\succ \dots \succ a_n$. For $j=2,3\dots m+1$ we define $\sigma_j$ be the same ordering as $\sigma_1$ but $\sigma_1(1)$ and $\sigma_1(j)$ swapped, i.e., $\sigma_j=a_j\succ a_2\succ a_3\succ \dots a_{j-1}\succ a_1\succ a_{j+1} \dots \succ a_n$.

First we bound  the probability that a given graph $\G$ is sampled from $\sigma_1$ versus its probability when sampled from $\sigma_j$: Note that the two permutations agree on all pairwise comparisons  expect $a_1$ with $a_\kappa$ and $a_j$ with $a_\kappa$ for all $\kappa\in 2,3,\dots j-1$ and $a_i,a_j$. Therefore, there are at most $2(j-2)+1=2j-3\leq 2m-1$ pairwise disagreements between the two permutations.

Consider a graph $\G$ and an edge in it. If the two permutation rank the end points of that edge in the same way the ratio of  the probability that it gets fixed under $\Ps(\sigma_a)$ and $\Ps(\sigma_j)$ as in $\G$ is the same. If they disagree, in one of them the edge will be sampled with probability $1/2-\gamma$ from one and with $1/2+\gamma$ from the other one. Using \cref{eq:fixprob} the probability that get fixed in the give direction w.r. to one matrix is $(1/2+\gamma)^2/[(1/2+\gamma)^2+(1/2-\gamma)^2]$ and w.r.t. another one it is $(1/2-\gamma)^2/[(1/2+\gamma)^2+(1/2-\gamma)^2]$. 

Therefore, denoting the probability that $\G$ is sampled from $\sigma_1$ by $\mathbb{P}_{\sigma_1}(\G)$ and the probability that  $\G$ is sampled from $\sigma_j$ by $\mathbb{P}_{\sigma_j}(\G)$, we have:

\[
\forall \G, \quad \left(\frac{1/2-\gamma}{1/2+\gamma}\right)^{2\cdot (2m-1)}\leq \frac{\mathbb{P}_{\sigma_1}(\G)}{\mathbb{P}_{\sigma_j}(\G)}\leq \left(\frac{1/2+\gamma}{1/2-\gamma}\right)^{2\cdot (2m-1)}
\]

Consider any arbitrary set $S$ with size $m$ and let $G$ be the set of all graphs $\G$ on which $S=Alg(\G)$. 
For any $j=1,2,\dots, m+1$ let $\mathbb{P}_{\sigma_j}(S)$ be the probability that $Alg$ outputs $S$. We have:
\[
\mathbb{P}_{\sigma_j}(S)=\sum_{\G\in G}\mathbb{P}_{\sigma_j}(\G)\geq \sum_{\G\in G}\big(\frac{1/2-\gamma}{1/2+\gamma}\big)^{2\cdot (2m-1)}\mathbb{P}_{\sigma_1}(\G)
=
\big(\frac{1/2-\gamma}{1/2+\gamma}\big)^{2\cdot (2m-1)} \mathbb{P}_{\sigma_1}(S)
\]

For any such output at least one of $a_1,a_2,\dots a_{m+1}$ is not in $S$. Thus, 
\begin{align*}
   \sum_{j=1}^{m+1} \mathbb{P}_{\sigma_j}\big(\sigma_j(1)\notin S\big)\geq 
   \big(\frac{1/2-\gamma}{1/2+\gamma}\big)^{2\cdot (2m-1)}
   \sum_{j=1}^{m+1} \mathbb{P}_{\sigma_1}\big(\sigma_j(1)\notin S\big)\geq    \big(\frac{1/2-\gamma}{1/2+\gamma}\big)^{2\cdot(2m-1)}~.
\end{align*}

Therefore, there is at least one $j$ such that 
\begin{align*}
  \mathbb{P}_{\sigma_j}\big(\sigma_j(1)\notin S\big)\geq    \big(\frac{1/2-\gamma}{1/2+\gamma}\big)^{2\cdot (2m-1)}/(m+1)~.
\end{align*}

Take $c_\gamma= \big(\frac{1/2-\gamma}{1/2+\gamma}\big)^5$ and we get the premise.
\end{proof}


\begin{corollary}
In our setting and by using only regular agents, and for any constant $C$, there is no algorithm  with failure probability bounded by $O(1/n)$ that outputs $S$  of size $m=o(\log n)$ containing $\top\left(\lceil m/C \rceil\right)\subseteq S$. 
Moreover,
there is no algorithm with failure probability bounded by $O(1/\log n)$ that outputs $S$  of size $m=o(\log \log n)$ containing $\top\left(\lceil m/C \rceil\right)\subseteq S$. 
\end{corollary}



From \Cref{lem:lowerbound1} we can also conclude that using only regular learning agents  no algorithm can output $\top(m)$ for constant $m$ and  guarantee a $o(1)$ failure probability.
In addition we show that following stronger lemma which has a similar proof presented in the appendix: 

\begin{restatable}{lemma}{lowerboundconstant}
\label{thm:lower:constant} Let  $0<k \leq n/2$ be a constant. 
For any given $0<\gamma<1/2$, consider an arbitrary algorithm $Alg$ as described above with output size   $k$. 
There exists constant $c_\gamma<1$ and  a hidden permutation $\sigma$  such that    with probability at least $c_{\gamma,k}$ the output of $Alg$, $S$ satisfies: $S\cap \top(k)=\emptyset$.
\end{restatable}

\subsection{Top-$k$ Retrieval with Stubborn Agents}\label{sec:Alg:private}

In this section, we study the top-$k$ retrieval problem in the regime where $k=\Theta(1)$. Recall that \Cref{lem:lowerbound1} and \Cref{thm:lower:constant} establish the impossibility of solving this problem using only 
regular learning agents. To overcome this limitation, we leverage \emph{stubborn} agents, whose actions help improve the accuracy of local estimates. 
Unlike regular agents, stubborn agents base their actions \emph{solely} on their own private signals. We denote the action of a stubborn agent $v_t$ conditioning on that nature sampled $(i,j)$ as $\act_{\stub}(v_t\mid i,j)$ and as $\act_{\stub}(v_t)$ without this condition.  Therefore, 
$ \act_{\stub}(v_t)=\sig(v_t)$. In our setting we also assume that stubborn agents know other stubborn agents.  
We invoke the procedure \compareii\ (\Cref{alg:compare2}) throughout our algorithm. Given a pair $i,j$ and a repetition parameter $\theta$, \compareii\ presents   $(a_i,a_j)$
to $\theta$ consecutive stubborn agents and outputs the inference made by the  agents who come after this sequence. We prove that:


\begin{restatable}{lemma}{stubborn}
\label{lem:herding:stubborn}
   Consider a pair $i<j\in [n]$, thus $a_i\succ a_j$ in the true order. Let $0<\delta<1$ be an arbitrary failure probability. For $\theta\geq \ln(\delta^{-1})/2\gamma^2$ we have that the outcome of \compareii\ is the true direction $a_i\succ a_j$ with probability at least $1-\delta$.
\end{restatable}

\begin{proof}
Under Assumption \ref{assumption},  we have that for any $i<j$, $p_{ij}-p_{ji}\geq 2\gamma$. Let $S_1$ be the number of times that an agent's private signal is $a_i\succ a_j$ and let $S_2$ denote the number of counts for the opposite signal. 
Note that $\mathbb{E}[S_1-S_2]\geq 2\gamma\cdot \theta$, thus, using Hoeffding bound we have:
\[
\mathbb{P}\left(
S_2\geq S_1
\right)\leq \mathbb{P}\left(
 S_1- S_2-\mathbb{E}[S_1-S_2] \leq -2 \theta\cdot \gamma 
\right)\leq \exp\left(
-2\gamma^2\theta
\right)\leq \delta~. 
\]
\end{proof}

A pseudocode of our proposed solution, \algstar\ (\textsc{STubborn Agent top-$k$ Retrieval}), is presented in \cref{alg:star}.

\begin{algorithm}
\caption{\algstar}\label{alg:star}
{\small
\begin{algorithmic}[1]

\STATE \textbf{Input:}  $\N$, $\gamma$, $\theta$, $\delta$, $k$ 
\STATE $\Gamma=\algone(\N,\gamma)$
\STATE \textbf{if} $|\Gamma|\leq k$ \textbf{then}, \textbf{return} $\Gamma$
\STATE $\theta=\ln(\vert \Gamma\vert^2\cdot \delta^{-1})/2\gamma^2$
\FOR{$i\in \Gamma$}
\STATE $\hat{R}(i)=0$
\FOR{$j\in \Gamma\setminus\{i\}$}
\STATE \textbf{if} 
$ \compareii(i,j)=a_j\prec a_i$ \textbf{then} $\hat{R}(i)++$
\ENDFOR
\ENDFOR
\STATE sort $\Gamma$ w.r.t. $\hat{R}$ and let $T$ be the top $k$ elements
\RETURN $T$
\end{algorithmic}
}
\end{algorithm}

Our proposed algorithm in this section is 
\algstar: \textsc{Stubborn Agent Top-$k$ Retrieval}. \algstar\ begins by invoking \algone\ to obtain a set of size $\Theta(\log n)$ that contains all top-$k$ elements. It then sets the repetition parameter $\theta$ for \compareii\ so that the estimated rank of each element is sufficiently close to its true rank. 
After estimating the ranks of all items output by \algone, the algorithms ranks them based on estimated ranks and outputs $k$ elements with highest ranks.

We note that the output of our previous algorithms \algone\ and \algtwo\ always contained some false positive errors. Here in contrast we show that \algstar\ outputs the top-$k$ list with high probability excluding all false positives and false negatives. 

\medskip 

The following theorem  presents theoretical guarantees for the performance of \algstar:

\begin{restatable}{theorem}{STAoR}
\label{thm:staralgo}
Fix $k$, for large enough $n$,
Algorithm \algstar\ receives set $\N$, parameters: $\gamma$,  $k$ and $\delta\geq 1/n$ as input, and it outputs a set $T$ of size $k$ such that,
$\mathbb{P}\left(S=\top(k)\right)\geq 1- 2\delta$.
Let $t_\mathrm{r}$ and $t_\mathrm{s}$ respectively denote the total number of regular and stubborn agents queried by \algstar\ to output $T$. We have that:
\[
\mathbb{E}[t_\mathrm{r}]\in O_\gamma\left(n\log^2(n)\right) \text{, and } t_\mathrm{s}\in O_\gamma\left(\log^2(n)[\log\log(n)+\log(\delta^{-1})]\right)~.
\]
\end{restatable}

\begin{proof}
Let $\Gamma$ be the output of \algone\ and 
$m=\vert \Gamma\vert$, if $2m/3\geq k$, then $\top(k)\subseteq \Gamma$.  Taking $\theta=\ln(m^{2}\cdot \delta^{-1})/8\gamma^2$, from  \cref{lem:herding:stubborn} the probability of any $i,j$ being in fixed in the wrong direction is bounded by $\delta/m^2$.
Taking a union bound over all $m^2$ pairs, the probability that any edge has the wrong direction is bounded by $m^2\cdot \delta/m^2\leq \delta$. 
Therefore, by sorting the elements w.r.t. estimated ranks $\hat{R}$ as calculated in \Cref{alg:star} and outputting the top $k$ elements with probability $1-\delta$ the two sets are equal. This completes the first part of the proof.

\medskip 

Note that 
\compareii\ is called $\Theta(m^2)$ times. 
Furthermore,  $m\leq L\leq  8*320^3/\gamma^5\log n$ (line 6 of \algone), thus, the total number of stubborn agents used by the algorithm is bounded by 
\[
\frac{m^2}{2\gamma^2}\ln(m^2/\delta)\in 
O_\gamma\left(\log^2(n)[\log\log(n)+\log(\delta^{-1})]\right)~.
\] 

Finally, we note that the algorithm is calling \algone\  whose success probability of is $1-1/n$. Since $\delta\geq 1/n$, the failure probability is bounded by
$
\delta+1/n\leq 2\delta 
$.
\end{proof}

\medskip 

We remark that there is a trade-off between the number of stubborn agents queried by \algstar\ and the algorithm's success probability: a success probability of $1 - O(1/n)$ is achievable using $t_{\text{s}} \in \Theta(\log^3 (n))$ stubborn agents, whereas a success probability of $1 - O(1/\log n)$ is achievable using $t_{\text{s}} \in \Theta(\log^2(n) \log \log (n))$ stubborn agents.

\section{Conclusion and Future Work}
In this paper, we study a sequential learning process in which agents' private signals are noisy private comparisons drawn from a true hidden ranking. We demonstrate that while herding poses challenges for inferring pairwise rankings, it is still possible to recover global information about the underlying order, and we develop algorithms to retrieve the highest-ranked elements of the true hidden order across different settings. 

A key modeling assumption in our work is that learning agents, when comparing a given pair, condition their inferences only on the actions of prior agents who evaluated that exact same pair. While this assumption is realistic when agents have limited information or cognitive capacities, relaxing it represents a compelling direction for future work. Specifically, a promising next step is to analyze the dynamics of local and global herding when agents incorporate triadic transitivity into their belief-updating processes.

\section{Acknowledgment}
The authors would like to thank the anonymous reviewers of WINE 2026 for their helpful feedback and suggestions.
We also acknowledge the use of GPT-6 Astra Pro for language proofreading, refining preliminary proofs, and gathering related literature. The authors retain full responsibility for all mathematical claims and final content.
Shahrzad Haddadan was supported by the Rutgers Business School Dean's Seed Funding Award. Jie Gao would like to acknowledge NSF
support through CNS-2515159, IIS-2229876, DMS-2220271, DMS-2311064, CCF-2118953.

\bibliographystyle{splncs04}
\bibliography{ref,learning}
\section{Appendix}

\begin{theorem}[\cite{feller1968}]
\label{thm:app:hitting}
Consider a biased random walk on the integer line which starts at $0$ and takes one step forward w.p. $p>1/2$ and backward  w.p. $1-p$. The walk  stops when hitting $k$ or $-k$. 
The probability that the walk hits $k$ before $-k$ is 

\begin{equation}
    \mathbb{P}\left( k \text{ is hit before }-k\right)=\frac{1}{1+\left(\frac{1-p}{p}\right)^k}~.
\end{equation}

Let $T_k$ be the number of steps this walk takes until it stops. We have:

\begin{equation}
    \mathbb{E}[T_k]=  
    \frac{k}{2p-1}\cdot \frac{1-(1-p/p)^k}{1+(1-p/p)^k}= \Theta\left(\frac{k}{2p-1}\right)~.
\end{equation}

\end{theorem}

\bigskip

\mugap*
\begin{proof}
 Consider $a_i\succ a_j\in S$. Let's first divide $S$ to $S_1$, $S_2$ and $S_3$  based on how they compare  $a_i$ and $ a_j$:
\[
S_1=\{a\in S\mid a\prec  a_j\}, 
S_2=\{a\in S\mid  a_j \preceq a\preceq a_i\},
S_3=\{a\in S\mid a_i\prec a\}
\]

Using the SST property we have that:
\begin{equation}\label{eq:sst}
    \forall a_\kappa \in S_1,~ p_{i\kappa}\geq p_{j \kappa} \text{ , and } \forall a_\kappa \in S_3,~  p_{ i \kappa}\geq p_{  j \kappa} 
\end{equation}

Note that since  $a_i\succ a_j$ we have 
$p^*_{ij}-p^*_{ji}\geq 2\gamma$ and for any 
 $a_k$ in between them we have: $p^*_{ik}-p^*_{jk}\geq 2\gamma$. 
Using  \cref{eq:sst} we have that 
$\mu_i(S_1)-\mu_j(S_1)\geq 0 $ and $\mu_i(S_3)-\mu_j(S_3)\geq 0$. Thus putting together we have:

\[
\mu_i(S_2)-\mu_j(S_2)\geq (p^*_{ij}-p^*_{ji})+ \sum_{a_\kappa\in S_2} \frac{p^*_{i\kappa}-p^*_{j\kappa}}{m-1} =2 \gamma \frac{\rank(a_j,S)-\rank(a_i,S)}{m-1}~.
\]
\end{proof}
\setdiflemma
\begin{proof}
Let $S$ be an arbitrary set satisfying the lemma's requirements. 
Note that for large enough $m_r$, $\top(\lfloor \frac{11}{16}m_r\rfloor)\subseteq S$, implies $\top(m_{r+1})\subseteq S$. We now show that  it is unlikely that \proc\ removes  many  of these elements. 
We define two sets:
\[
D=\top(m_{r+1})\setminus S', \quad K=S'\setminus \top(m_{r+1})~,
\]
where $D$ contains all \emph{the elements that we want them but the algorithm Discards them} and $K$ contains \emph{the elements that we don't want them but the algorithm Keeps them}. Note that since $|S'|=m_{r+1}$, we have $|K|=|D|$, assume that the size of both sets are $h$. 
Sort elements of $K$ and $D$ w.r. to their rank in $S$. It is not hard to see that the difference between the rank of the $i$th item from $D$, say $a_j$ and the $i$th item from $K$, say $a_{j'}$ satisfies: 
\[
\rank(a_{j'},S)- \rank(a_j,S)\geq  (m_{r+1}+i)-(m_{r+1}-h+i)\geq h~.\]

Therefore, using \cref{lem:mugap} and summing over $h$ pairs, we have: 
\begin{equation}\label{eq:10}
 \sum_{a_j\in D}\mu_j(S)- \sum_{a_i\in K}\mu_i(S) \geq \frac{2\gamma}{m_r-1} h^2\enspace .
\end{equation}

Note that the elements in $D$ must have a lower normalized  estimated  rank that those $K$ otherwise they won't be deleted by \proc, i.e., 
\begin{equation}\label{eq:11}
    \sum_{a_i\in K}\hat{\mu}_i(S)\geq 
     \sum_{a_j\in D}\hat{\mu}_j(S)~.
\end{equation}

We define:
\[
Er_K=\sum_{a_i\in K}\hat{\mu}_i(S)-\mu_i(S),\quad Er_D=\sum_{a_i\in D}\hat{\mu}_i(S)-\mu_i(S)~. 
\]

Note that if 
\begin{equation}\label{upperlower}
    Er_K< \frac{\gamma}{m_r-1}h^2 \text{ and } Er_D>-\frac{\gamma}{m_r-1}h^2
\end{equation} then from \cref{eq:10} we can conclude that \cref{eq:11} \emph{does not} hold. 
We now show that for $h\geq \gamma/320\cdot m_r$, the probability that \cref{upperlower} holds is large. 

Consider fixed $S$, $M$ and $D$, using \cref{eq:correctionterm} we write:
\[Er_K=
\sum_{a_i\in K}\hat{\mu}_i(S)-\mu_i(S)=
\sum_{a_i\in K} 
\hat{\mu}_i(S)-\mathbb{E}[\hat{\mu}_i(S)]+\sum_{a_i\in K} E_{i}(S)
~,
\]
A similar equation holds for a sum over all $Er_D$. Since we assume $S,D,K$ are all fixed sets, both of the above terms have expected value zero and we can apply Hoeffding bound as follows:

We normalize the first term and bound the empirical average over   $hD_r$ different random variables using Hoeffding bound:
\[
\mathbb{P}
\left(\frac{1}{h}
\sum_{a_i\in K} \left[
\hat{\mu}_i(S)-\mathbb{E}[\hat{\mu}_i(S)]\right]\geq \frac{\gamma}{2(m_r-1)}h
\right)\leq \exp\left(
- 2 \frac{\gamma^2h^3}{4(m_r-1)^2}D_r
\right)~.
\]
Similarly for $D$ we can show:
\[
\mathbb{P}
\left(\frac{1}{h}
\sum_{a_j\in D} 
\hat{\mu}_j(S)-\mathbb{E}[\hat{\mu}_j(S)]\leq -\frac{\gamma}{2(m_r-1)}h
\right)\leq \exp\left(
- 2 \frac{\gamma^2h^3}{4(m_r-1)^2}D_r
\right)~.
\]

In order to bound the second terms of $Er_K$ and $Er_D$ let's write:
\begin{align*}
   (|S|-1) \sum_{a_i\in K} E_i(S)= \sum_{a_i\in K}\sum_{a_{i'}\in S}E_{ii'}= \sum_{a_i\in K}\sum_{a_{i'}\in S\setminus K}E_{ii'}+\sum_{a_i\in K} \sum_{a_{i'}\in K}E_{ii'}
\end{align*}

Note that since $E_{ii'}=-E_{i'i}$ we have:
\begin{align*}
(|S|-1) \sum_{a_i\in K} E_i(S)= \sum_{a_i\in K}\sum_{a_{i'}\in S\setminus K}E_{ii'}
\end{align*}

Therefore, the term 
$(|S|-1)\sum_{a_i\in K}E_i(S)$  is a sum of  $h(m_r-h)$ independent  random variables. Thus:
\begin{align*}
    \mathbb{P}
\left(
\sum_{a_i\in K} 
E_i(S)\geq \frac{\gamma}{2(m_r-1)}h^2
\right)&= \mathbb{P}
\left((|S|-1)
\sum_{a_i\in K} 
E_i(S)\geq \frac{\gamma}{2}h^2
\right)\\
&\leq  \exp\left(
- 2 \frac{\gamma^2h^4}{4}\frac{1}{h(m_r-h)}
\right)\leq 
\exp\left(
- 2 \frac{\gamma^2h^3}{4(m_r-h)}
\right)~.
\end{align*}

Similarly, we can show that for $D$:
\begin{align*}
    \mathbb{P}
\left(
\sum_{a_j\in D} 
E_j(S)\leq -\frac{\gamma}{2(m_r-1)}h^2
\right)\leq 
\exp\left(
- 2 \frac{\gamma^2h^3}{4(m_r-h)}
\right)~.
\end{align*}

 Take  $h\geq \gamma/320\cdot m_r$. Using union bound the probability that \cref{upperlower} does not hold  is bounded by: 
\begin{align*}
       p&= 2\exp\left(
-\frac{\gamma^2}{2m_r^2}\left(\frac{\gamma m_r}{320}\right)^3 D_r
\right)+ 
2\exp\left(
-\frac{\gamma^2}{2(m_r-h)}\left(\frac{\gamma m_r}{320}\right)^3 
\right)~ \\
&\leq  
2\exp\left(
-\frac{\gamma^5}{2*320^3}\cdot m_r C_\gamma \log(m_r\delta^{-1}) 
\right)
+2\exp\left(
-\frac{\gamma^5}{2*320^3}\cdot m^2_r 
\right)
~ \\
&\leq  
2\exp\left(
- 2 m_r \log(m_r\delta^{-1}) 
\right)
+2\exp\left(
-\frac{\gamma^5}{2*320^3}\cdot m^2_r 
\right)
~ \\
\end{align*}

Putting together, for fixed $S, M$, and $D$. The probability that both \cref{eq:10} and \cref{eq:11} hold is less than $p$. Note that $S,D$ and $M$ are not fixed set. However, since $S\subseteq \top(2m_r)$ and so are $D$ and $M$. Therefore, an upper bound for the total number of choices for selecting these sets is $2^{6m_r}$. Taking a union bound over all such sets and plugging in the above bound for $p$  we have: 
\begin{align}\label{secondbound}
      & \mathbb{P}\left(\big|
\top(m_{r+1})\setminus S'\big|\geq \frac{\gamma}{320}\cdot m_r
\right)\leq  2^{6m_r+2}p\nonumber\\
\leq  &
2\left[\exp\left(
(6m_r+2)\ln(2)- 2 m_r \log(m_r\delta^{-1}) 
\right)
+\exp\left((6m_r+2)\ln(2)
-(2/C_\gamma)\cdot m^2_r 
\right)\right]
\end{align}

The first term is bounded by $\delta/m_r$ for large enough $m_r$. We bound the second term  using 
$\delta\geq n^{-2}$, and $m_{r}\geq 4*320^3/\gamma^5
\log n$.
First note that 
\[
\exp\left((6m_r+2)\ln(2)
-(2/C_\gamma)\cdot m^2_r 
\right)\in O\left(\exp\left(-(1/C_\gamma)\cdot m^2_r 
\right)\right)  \]
Using 
$\delta\geq n^{-2}$, and $m_{r}\geq 4*320^3/\gamma^5
\log n$, we have:
\begin{align*}
&\exp\left(-(1/C_\gamma)\cdot m^2_r 
\right)
 \leq   \exp\left(-(1/C_\gamma)\cdot (C_\gamma \log n)^2 
\right) \\
=&  \exp\left(- C_\gamma \log^2 n 
\right)\leq (1/n^2)^{C_\gamma}\leq \delta/m_r~.
\end{align*}
Therefore, the failure probability is at most $4\delta/m_r$, which is the premise.  
\end{proof}

\thmalgtwo* 
\begin{proof}

Let us first bound the probability of failure for \algtwo\. 
Denote $C_\gamma=8*320^3/\gamma^5$, $L_1=C_\gamma\log n$ and $L_2=C_\gamma\log\log n +C_\gamma\log C_\gamma$. 

Note that \algtwo{} has two  stages: first it calls \algone{} which has probability of failure at most $1/n$. Then it keeps calling \proc\ with initial set $|\Gamma|\leq L_1$, and $\delta=1/\log n$. Let $\mathbf{r}$ be the final round of \algtwo.  Inorder to  obtain the proof, we need a new inductive argument which is presented in \cref{inductiveproofsecondthm}. employing it  to $\mathbf{r}$ we have: 
\begin{align*}
    \mathbb{P}\left(
\top(\lfloor\frac{11}{16} m_{\mathbf{r}}\rfloor )\not\subset \Gamma_\mathbf{r}
\right)&\leq \frac{9\delta}{L_1}\sum_{r=0}^\mathbf{r}(3/2)^r& \text{Statement 1 of \cref{inductiveproof} }\\
&\leq \frac{6}{C_\gamma \log n}\frac{\log \log n}{L_2}& \mathbf{r}\leq \log \log n, \delta=1/\log n, m_r\geq L_2\\
&\leq 1/ 2\log n~.
\end{align*}

The proof for bounding the query complexity follows analogously to the proof of \cref{thm:fistalg} using $\mathbf{r}=\log \log n$ and observing that the sample complexity of each iteration of \proc\ is bounded by $\Theta(m_r \log(m_r\log n ))$.

\end{proof}

\begin{theorem}\label{inductiveproofsecondthm}
Let $\Gamma_0=\N$ and for $r\in \{0,1,2,\dots, \mathbf{r}-1\}$ we have: $\Gamma_{r+1}=\proc(\Gamma_r)$. For $r\in \{0,1,\dots, \mathbf{r}\}$ we denote $|\Gamma_r|=m_r$. Given $\delta\geq \log n^{-2}$, and $C_\gamma \log n\geq m_{r}\geq C_\gamma\log\log n +C_\gamma\log C_\gamma$, where $C_\gamma=4*320^3/\gamma^5$. At round $r$ it holds that:
\begin{align*}
 1.\quad & \top(\lfloor\frac{11}{16}m_r\rfloor)\subseteq \Gamma_r\subseteq \top( 2 m_r)~.\\
 2.\quad & \big|  \Gamma_r\setdiff\top(m_r) \big|\leq \frac{\gamma}{80}(m_r-1)~.
\end{align*}
with failure probability at most  
\[
\frac{9\delta}{n}\sum_{i=0}^{r} (\frac{3}{2})^i~.
\]
\end{theorem}
\begin{proof}
The proof is very similar to the proof of \cref{inductiveproof}. There is a slight modification in derivation of \cref{firstbound} and \cref{secondbound} which we show here:

For $\delta\geq \log n^{-2}$, and $C_\gamma\log n \geq m_{r}\geq C_\gamma\log\log n +C_\gamma\log C_\gamma$, where $C_\gamma=4*320^3/\gamma^5$, we have: 

\begin{align}\label{firstbound}
   2m_r\exp\left(-\frac{\gamma^2}{6400}\cdot m_r\right)&\leq 2m_r\exp\left(-\frac{\gamma^2}{6400}\cdot \frac{4*320^3}{\gamma^5}\log \log n\right) \nonumber\\
   &\leq 2m_r\exp\left(-\frac{4*512}{\gamma^3}\log \log n\right) \nonumber
   \\
   &\leq 2m_r\exp\left(-256\log \log n\right) \leq \frac{2m_r}{(\log n) ^{256}}\leq \delta/m_r~.
\end{align}

Furthermore, for $\delta\geq \log n^{-2}$, and $m_{r}\geq C_\gamma\log\log n +C_\gamma\log C_\gamma$, where $C_\gamma=4*320^3/\gamma^5$, we have: 
\begin{align*}
&2\left[\exp\left(
(6m_r+2)\ln(2)- 2 m_r \log(m_r\delta^{-1}) 
\right)
+\exp\left((6m_r+2)\ln(2)
-(2/C_\gamma)\cdot m^2_r 
\right)\right]\\
&\leq 4\delta/m_r~.
\end{align*}
\end{proof}

\lowerboundconstant*
\begin{proof}
The proof is similar to the proof of previous lemma with the following modification:

Define $\sigma_1=a_1\succ a_2\succ \dots \succ a_k \succ a_{k+1}\succ\dots\succ a_n$. Let $L={{2k}\choose{k}}$. 
Consider all subsets of $a_1,a_2,\dots, a_{2k}$ and order them arbitrarily. 
For $j=2,3,\dots, L$ define: $\sigma_j$ be the permutation where the subset with index $j$ is placed at $\top_{\sigma_j}(k)$ preserving the order in $\sigma_1$ and all other elements are shifted right to preserve their relative order. Note that for any such $\sigma_j$ the number of pairwise disagreements of $\sigma_1$ and $\sigma_j$ at at most $k^2$. Furthermore, if the size of output $S$ is $k$, at least one of $\sigma_j$s will satisfy $S\cap \top(\sigma_j)=\emptyset $. The rest of the proof following similarly to the proof of \cref{lem:lowerbound1}, and we can show that there is at least one $j$ such that 

\begin{align*}
\mathbb{P}_{\sigma_j}\big(\top_{\sigma_j}(k)\cap S=\emptyset\big)\geq    \big(\frac{1/2-\gamma}{1/2+\gamma}\big)^{2\cdot (k^2)}/{{2k}\choose{k}}~.
\end{align*}

The proofs concludes from setting $c_{\gamma,k}=\big(\frac{1/2-\gamma}{1/2+\gamma}\big)^{2\cdot(k^2)}/{{2k}\choose{k}}$. 
\end{proof}

\end{document}